\documentclass[submission,copyright,creativecommons]{eptcs}
\providecommand{\event}{NCMA 2024} 

\usepackage[english]{babel}
\usepackage{iftex}
\usepackage{hyperref}
\usepackage{tikz,pgf}
\usetikzlibrary{positioning, arrows, automata}
\usepackage{comment}
\usepackage{amsmath, amsfonts, amssymb, theorem, enumerate}

\newcommand{\N}{\ensuremath{\mathbb{N}}}

\newcommand{\NIL}{\mathit{NIL}}
\newcommand{\COMB}{\mathit{COMB}}
\newcommand{\DEF}{\mathit{DEF}}
\newcommand{\SYDEF}{\mathit{SYDEF}}
\newcommand{\COMM}{\mathit{COMM}}
\newcommand{\CIRC}{\mathit{CIRC}}
\newcommand{\SUF}{\mathit{SUF}}

\newcommand{\NC}{\mathit{NC}}
\newcommand{\SF}{\mathit{SF}}
\newcommand{\UF}{\mathit{UF}}
\newcommand{\FIN}{\mathit{FIN}}
\newcommand{\MON}{\mathit{MON}}
\newcommand{\PS}{\mathit{PS}}
\newcommand{\ORD}{\mathit{ORD}}
\newcommand{\STAR}{\mathit{STAR}}
\newcommand{\COM}{\mathit{COM}}
\newcommand{\LCOM}{\mathit{LCOM}}
\newcommand{\RCOM}{\mathit{RCOM}}
\newcommand{\TCOM}{\mathit{2COM}}
\newcommand{\REG}{\mathit{REG}}
\newcommand{\Suf}{\mathit{Suf}}
\newcommand{\ec}[1]{\ensuremath{\mathcal{EC}(#1)}}

\def\Set#1#2{\left\{\: #1\;|\; #2\:\right\}}

\def\Sets#1{\left\{\,#1\,\right\}}
\def\set#1#2{\{\; #1 \mid #2\;\}}
\def\sets#1{\{#1\}}

\def\cC{{\cal C}}
\def\cF{{\cal F}}

\def\cR{{\cal R}}
\def\cS{{\cal S}}
\def\cT{{\cal T}}
\def\cEC{{\cal EC}}
\DeclareSymbolFont{symbols}{OMS}{cmsy}{m}{n}

\def\Lra{\Longrightarrow}

\newtheorem{theorem}{Theorem}
\newtheorem{lemma}[theorem]{Lemma}
\newtheorem{corollary}[theorem]{Corollary}
\theorembodyfont{\rm}
\newtheorem{example}[theorem]{Example}
\newenvironment{proof}{{\em Proof. }}{{}\hspace*{\fill}$\Box$ \par \medskip }
\newenvironment{proof*}{{\em Proof. }}{\par \medskip }

\newlength{\btlabelwidth}
\newlength{\btleftmargin}
\newenvironment{btlists}{\begin{list}{{\rm--}}{%
\setlength{\labelwidth}{\btlabelwidth}\setlength{\leftmargin}{\btleftmargin}%
\setlength{\topsep}{0pt plus0.2ex}%
\setlength{\itemsep}{0ex plus0.2ex}%
\setlength{\parsep}{0pt plus0.2ex}}}{\end{list}}
\tikzstyle{to}=[->, >=stealth]
\tikzstyle{hier}=[->, >=angle 60]
\tikzstyle{hiero}=[->, >=angle 60, dashed]
\tikzstyle{state}=[circle,draw,inner sep=2pt,minimum size=8mm]
\tikzstyle{edgeLabel}=[inner sep=0.5mm,fill=white,text=black]

\usepackage{pdflscape}
\ifpdf
  \usepackage{underscore}     
  \usepackage[T1]{fontenc}    
\else
  \usepackage{breakurl}       
\fi

\title{Various Types of Comet Languages and their Application in External Contextual Grammars}
\author{Marvin K\"odding \qquad \qquad Bianca Truthe
\institute{Institut f\"ur Informatik, Universit\"at Giessen\\Arndtstr. 2, 35392 Giessen, Germany}
\email{\{marvin.koedding,bianca.truthe\}@informatik.uni-giessen.de}
}
\def\titlerunning{Various Types of Comet Languages and their Application in External Contextual Grammars}
\def\authorrunning{M. K\"odding \& B. Truthe}

\begin{document}
\maketitle

\begin{abstract}
In this paper, we continue the research on the power of contextual grammars with selection languages from 
subfamilies of the family of regular languages. We investigate various comet-like types of languages and compare
such language families to some other subregular families of languages (finite, monoidal, nilpotent, combinational, 
(symmetric) definite, ordered, non-counting,  power-sepa\-rating,
suffix-closed, commutative, circular, or union-free languages). Further, we compare the language families defined
by these types for the selection with each other and with the families of the hierarchy obtained for external
contextual grammars. In this way, we extend the existing hierarchy by new language families. 
\end{abstract}


\section{Introduction}

Contextual grammars were first introduced by Solomon Marcus in \cite{Marcus.1969} as a formal model that might be
used for the generation of natural languages. The derivation steps
consist in adding contexts to given well-formed sentences, starting from an initial finite basis. Formally, a context is given by a pair $(u, v)$ of words and the external adding to a word $x$ gives the word $uxv$.  In order to control the derivation
process, contextual grammars with selection in a certain family of languages were defined. In such contextual grammars, a context $(u, v)$ may be added only around a word $x$ if this word $x$ belongs to a language which is associated with
the context. Language families were defined where all selection languages in a contextual grammar belong to some language family $F$. 

The study of external contextual grammars with selection in special regular sets was started by
J\"urgen Dassow in \cite{Dassow.2005}. The research was continued by J\"urgen Dassow, Florin Manea, 
and Bianca Truthe (see~\cite{Dassow_Manea_Truthe.2012}) where further subregular families of selection 
languages were considered.

In the present paper, 
we extend the hierarchy of subregular language families by families of comet-like languages. 
Furthermore, we investigate the generative capacity of external contextual grammars with  
selection in such subregular language families.


\section{Preliminaries}

Throughout the paper, we assume that the reader is familiar with the basic concepts of the theory of automata and formal languages. For details, we refer to \cite{Rozenberg_Salomaa.1997}. Here we only recall some notation, definitions, and previous results which we need for the present research.

An alphabet is a non-empty finite set of symbols. For an alphabet $V$, we denote by $V^*$ and $V^+$ the set of all 
words and the set of all non-empty words over $V$, respectively. The empty word is denoted by~$\lambda$. 
For a word $w$ and a letter $a$, we denote the length of $w$ by $|w|$ and the number
of occurrences of the letter~$a$ in the word $w$ by $|w|_a$. For a set $A$, we denote its cardinality by $|A|$.
The reversal of a word $w$ is denoted by $w^R$: if $w=x_1x_2\ldots x_n$ for letters $x_1,\ldots,x_n$, 
then $w^R=x_nx_{n-1}\ldots x_1$. By $L^R$, we denote the language of all reversals of the words in $L$: 
$L^R=\set{w^R}{w\in L}$.

A deterministic finite automaton is a quintuple 
\[{\cal A}=(V,Z,z_0,F,\delta)\]
where $V$ is a finite set of input symbols, $Z$
is a finite set of states, $z_0\in Z$ is the initial state, $F\subseteq Z$ is a set of accepting states, and $\delta$ is
a transition function $\delta: Z\times V\to Z$. The language accepted by such an automaton is the set of all input words 
over the alphabet $V$ which lead letterwise by the transition function from the initial state to an accepting state.

A regular expression over some alphabet $V$ is defined inductively as follows:
\begin{enumerate}
\item $\emptyset$ is a regular expression;
\item every element $x\in V$ is a regular expression;
\item if $R$ and $S$ are regular expressions, so are the concatenation $R\cdot S$,
the union~$R\cup S$, and the Kleene closure $R^*$;
\item for every regular expression, there is a natural number $n$ such that the 
regular expression is obtained from the atomic elements $\emptyset$ and $x\in V$
by $n$ operations concatenation, union, or Kleene closure.
\end{enumerate}

The language $L(R)$ which is described by a regular expression $R$ is also 
inductively defined:
\begin{enumerate}
\item $L(\emptyset)=\emptyset$;
\item for every element $x\in V$, we have $L(x)=\sets{x}$;
\item if $R$ and $S$ are regular expressions, then 
\[L(R\cdot S) = L(R)\cdot L(S),\quad
L(R\cup S)  = L(R)\cup L(S), \quad
L(R^*)    = (L(R))^*.
\]
\end{enumerate}

A general regular expression admits as operations (in the third item of the definition above) also 
intersection (where $L(R\cap S)=L(R)\cap L(S)$) and complementation (where~$L(\overline{R})=\overline{L(R)}$).

All the languages accepted by a finite automaton or described by some regular expression are called regular and form a
family denoted by $\REG$. Any subfamily of this set is called a subregular language family.

\subsection{Some subregular language families}

We consider the following restrictions for regular languages. In the following list of properties, we give already the abbreviation which denotes the family of all languages with the respective property. Let $L$ be a regular language over an alphabet $V$. With respect to the alphabet $V$, the language $L$ is said to be
\begin{itemize}
\item \emph{monoidal} ($\MON$) if and only if $L=V^*$,
\item \emph{nilpotent} ($\NIL$) if and only if it is finite or its complement $V^*\setminus L$ is finite,
\item \emph{combinational} ($\COMB$) if and only if it has the form
$L=V^*X$
for some subset $X\subseteq V$,
\item \emph{definite} ($\DEF$) if and only if it can be represented in the form
$L=A\cup V^*B$
where~$A$ and~$B$ are finite subsets of $V^*$,
\item \emph{symmetric definite} ($\SYDEF$) if and only if $L = EV^*H$ for some regular languages $E$ and $H$,
\item \emph{suffix-closed} ($\SUF$) (or \emph{fully initial} or \emph{multiple-entry} language) if
and only if, for any two words over $V$, say $x\in V^*$ and~$y\in V^*$, the relation $xy\in L$ implies
the relation~$y\in L$,
\item \emph{ordered} ($\ORD$) if and only if the language is accepted by some deterministic finite
automaton 
\[{\cal A}=(V,Z,z_0,F,\delta)\]
with an input alphabet $V$, a finite set $Z$ of states, a start state $z_0\in Z$, a set $F\subseteq Z$ of
accepting states and a transition mapping $\delta$ where $(Z,\preceq )$ is a totally ordered set and, for
any input symbol~$a\in V$, the relation $z\preceq z'$ implies $\delta (z,a)\preceq \delta (z',a)$,
\item \emph{commutative} ($\COMM$) if and only if it contains with each word also all permutations of this
word,
\item \emph{circular} ($\CIRC$) if and only if it contains with each word also all circular shifts of this
word,
\item \emph{non-counting} ($\NC$) if and only if there is a natural
number $k\geq 1$ such that, for any three\linebreak\ words~$x\in V^*$, $y\in V^*$, and $z\in V^*$, it 
holds~$xy^kz\in L$ if and only if $xy^{k+1}z\in L$,
\item \emph{star-free} ($\SF$) if and only if $L$ can be described by a regular expression which is built by concatenation, union, and complementation,     
\item \emph{power-separating} ($\PS$) if and only if, there is a natural number $m\geq 1$ such that
for any word~$x\in V^*$, either
$J_x^m \cap L = \emptyset$
or
$J_x^m\subseteq L$
where
$J_x^m = \set{ x^n}{n\geq m}$,
\item \emph{union-free} ($\UF$) if and only if $L$ can be described by a regular expression which
is only built by concatenation and Kleene closure,
\item \emph{star} ($\STAR$) if and only if $L = H^*$ for some regular language $H \subseteq V^*$,
\item \emph{left-sided comet} ($\LCOM$) if and only if $L = EG^*$ for some regular language $E$ and a regular language $G \notin \{\emptyset, \{\lambda\}\}$,
\item \emph{right-sided comet} ($\RCOM$) if and only if $L = G^*H$ for some regular language $H$ and a regular language $G \notin \{\emptyset, \{\lambda\}\}$,
\item \emph{two-sided comet} ($\TCOM$) if and only if $L = EG^*H$ for two regular languages $E$ and $H$ and a regular language $G \notin \{\emptyset, \{\lambda\}\}$.
\end{itemize}

We remark that monoidal, nilpotent, combinational, (symmetric) definite, ordered, non-counting, star-free, union-free, star, and (left-, right-, or two-sided) comet languages are regular, whereas non-regular languages of the other types mentioned above exist.
Here, we consider among the suffix-closed, commutative, circular, 
and power-separating languages only those which are also regular.
By $\FIN$, we denote the family of languages with finitely many words.
In \cite{McNaughton_Papert.1971}, it was shown that the families of the non-counting languages and 
the star-free languages are equivalent ($\NC=\SF$).

Some properties of the languages of the classes mentioned above can be found in
\cite{Shyr.1991} (monoids),
\cite{Gecseg_Peak.1972} (nilpotent languages),
\cite{Havel.1969} (combinational and commutative languages),
\cite{Perles_Rabin_Shamir.1963} (definite languages),
\cite{Paz_Peleg.1965} (symmetric definite languages),
\cite{Gill_Kou.1974} and \cite{Brzozowski_Jiraskova_Zou.2014} (suffix-closed languages),
\cite{Shyr_Thierrin.1974.ord} (ordered languages),
\cite{Kudlek.2004} (circular languages),
\cite{McNaughton_Papert.1971} (non-counting and 
star free
languages),
\cite{Shyr_Thierrin.1974.ps} (power-separating languages),
\cite{Brzozowski.1962} (union-free languages),
\cite{Brzozowski.1967} (star languages),
\cite{Brzozowski_Cohen.1969} (comet languages).

%
%
%

\subsection{Contextual grammars}

Let $\cF$ be a family of languages. A contextual grammar with selection in $\cF$ is a triple
$G=(V,\cS,A)$
where
\begin{btlists}
\item $V$ is an alphabet, 
\item $\cS$ is a finite set of selection pairs $(S,C)$ with a selection language $S$ over some subset $U$ of the
alphabet $V$ which belongs to the family $\cF$ with respect to the alphabet $U$ and a finite
set~\hbox{$C\subset V^*\times V^*$} of contexts where, for each context $(u,v)\in C$, at least one side is not empty: $uv\not=\lambda$,
\item $A$ is a finite subset of $V^*$ (its elements are called axioms).
\end{btlists}
We write a selection pair $(S,C)$ also as $S\to C$. In the case that $C$ is a singleton set $C=\{(u,v)\}$, we
also write $S\to(u,v)$.
For a contextual grammar 
$G=(V,\Sets{(S_1,C_1),(S_2,C_2),\dots ,(S_n,C_n)},A)$,
we set
\[\ell_A(G) = \max \Set{ |w| }{ w\in A },\quad
\ell_C(G) = \max \Set{ |uv| }{(u,v)\in C_i, 1\leq i\leq n},\quad
\ell(G)   = \ell_A(G)+\ell_C(G)+1.\]
We now define the derivation modes for contextual grammars with selection.

Let $G=(V,\cS,A)$ be a contextual grammar with selection.
A direct external derivation step in $G$ is defined as follows: a word~$x$ derives a word $y$ 
(written as~$x\Lra y$) if and only if there is a pair~$(S,C)\in\cS$ such that~$x\in S$ and $y=uxv$ 
for some pair $(u,v)\in C$.
Intuitively, one can only wrap a context $(u,v)\in C$ around a word $x$ if $x$ belongs to the corresponding
selection language $S$.



By $\Lra^*$ we denote the reflexive and transitive closure of the relation~$\Lra$. 
The language generated by $G$ is 
$L=\set{ z }{ x\Lra^* z \mbox{ for some } x\in A }$.

\begin{example}\label{ex-cg}
Consider the contextual grammar
$G=(\sets{a,b,c},\sets{(S_1,C_1),(S_2,C_2)},\sets{\lambda})$
with
\[S_1 = \sets{a,b}^*,\quad C_1 = \sets{(\lambda,a),(\lambda,b)},\qquad 
  S_2 = \sets{ab}^*,\quad C_2 = \sets{(c,c)}.\]

Starting from the axiom $\lambda$, every word of the language $S_1$ is generated by applying the first selection. 
Starting from any word of $S_2\subset S_1$, every word of the language $\{c\}S_2\{c\}$ is generated by 
applying the second selection. Other words are not generated.

Thus, the language generated is
\[L(G)=\sets{a,b}^*\cup\set{c(ab)^nc}{n\geq 0}.\]

Both selection languages are ordered: The language $S_1$ is accepted by a finite automaton with exactly one state. Hence, it is ordered.
The language $S_2$ is accepted by the following deterministic finite automaton 
$A = (\{z_0, z_1, z_2, z_3\}, \{a,b\}, \delta, z_1, \{z_1\})$ where the transition function
is illustrated in the following picture and given in the table next to it, from which it can 
be seen that the automaton is ordered:

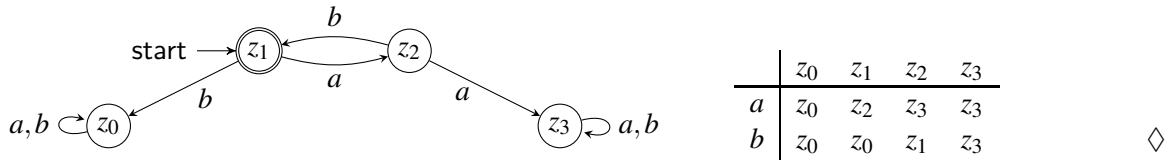
\begin{tikzpicture}[on grid,>=stealth',initial text={\sf start}
]
\node[state,minimum size=4mm,initial,accepting] (z_1) at (2,1) {$z_1$};
\node[state,minimum size=4mm] (z_2) at (4,1) {$z_2$};
\node[state,minimum size=4mm] (z_3) at (6,0) {$z_3$};
\node[state,minimum size=4mm] (z_0) at (0,0) {$z_0$};
\draw[to] 
(z_1) edge [bend right=15] node [below] {$a$} (z_2)
(z_2) edge [bend right=15] node [above] {$b$} (z_1)
(z_1) edge node [below,pos=.3] {$b$} (z_0)
(z_2) edge node [below,pos=.3] {$a$} (z_3)
(z_3) edge [loop right] node [right] {$a,b$} (z_3)
(z_0) edge [loop left] node [left] {$a,b$} (z_0);
\end{tikzpicture}
\qquad
\begin{tabular}[b]{c|cccc}
     & $z_0$ & $z_1$ & $z_2$ & $z_3$\\
     \hline
     $a$ & $z_0$ & $z_2$ & $z_3$ & $z_3$\\
     $b$ & $z_0$ & $z_0$ & $z_1$ & $z_3$
\end{tabular}
\hspace*{\fill}{$\Diamond$}
\end{example}

By~$\cEC(\cF)$, we denote the family of all languages generated externally by contextual grammars 
with selection in $\cF$. When a contextual grammar works in the external mode, we call it an external 
contextual grammar. 

The language generated by the external contextual grammar in Example~\ref{ex-cg} belongs, for instance, 
to the family $\cEC(\ORD)$ because all selection languages ($S_1$ and $S_2$) are ordered.

\section{Results on families of comet languages}

We first present some observations about star languages and two-sided comet languages, 
we give normal forms for two-sided comets, and we insert the subregular families investigated here 
into the existing hierarchy. 

From the structure of two-sided comet languages (languages $L$ of the form $EG^*H$ where $G$ is neither the empty set nor the set with the empty word only), we see that every such language is infinite if none of the sets $E$, $G$, and $H$ is the empty set. If one of the sets $E$ or $H$ is empty, then the whole language $L$ is also empty.

\begin{lemma}\label{lemma:2com_unendlich_oder_leereSprache}
  For each language $L \in \TCOM$, it holds that $L$ is either infinite or empty.
\end{lemma}
%

A similar observation can be made for star languages.

\begin{lemma}\label{lemma:star_unendlich_oder_leerwort}
  For each language $L \in \STAR$, it holds that $L$ either is infinite or consists of the empty word $\lambda$.
\end{lemma}
%
%
%

\subsection{Normal forms}

  We first show some observations before we conclude a normal form for languages from the class $\TCOM$. This normal form is later used when we prove that $\TCOM$-languages as selection languages are as powerful as arbitrary regular languages.
  
  \begin{lemma}\label{lemma:2com_zerlegung}
  Each two-sided comet language $L = EG^*H$ can be represented as a finite union 
  \[L = \bigcup_{i=1}^{n} E_i G^* H\]
  for some number $n \geq 1$ and with union-free languages $E_i$ for all $1 \leq i \leq n$.
  \end{lemma}
  \begin{proof}
  Let $L = EG^*H$ be a two-sided comet language.
  Every regular language is the union of finitely many union-free languages \cite{Nagy.2019}. Let
  $n\geq 1$ be a natural number and $E_i$ be a union-free language for any $i$ with $1\leq i\leq n$ such that 
  $E = E_1\cup E_2 \cup \cdots \cup E_n$.
  Then, it follows $L = E_1G^*H \cup E_2G^*H \cup \cdots \cup E_nG^*H$.
  \end{proof}
  
  In order to show later that we can transform any $\TCOM$-language into the mentioned normal form, we now present how an infinite union-free language can be represented by a special $\TCOM$-form.
  
  \begin{lemma}\label{lemma:reg_aufteilung}
  For an infinite union-free language $L$, there exist sets $L_l$, $L_i$, and $L_r$ such that~$L = L_lL_i^*L_r$ where $L_l$ is finite and $L_i \notin \{\emptyset, \{\lambda\}\}$.
  \end{lemma}
  \begin{proof}
  We prove the assertion inductively via the number of construction steps required to create a regular expression $\mathcal{R}$ such that $L=L(\mathcal{R})$ holds. In construction step 0, only finite languages are created. Therefore, the base case is $n=1$.\smallskip
  
  \noindent{\sl Base case $n=1$}:
  Since $L$ is infinite, we have $\mathcal{R} = \{x\}^*$ for a letter $x\in V$. A desired representation for the language $L$ is then $\{\lambda\}\{x\}^*\{\lambda\}$.\smallskip

  \noindent{\sl Induction step $n\to n+1$}:
  Assume the induction hypothesis: For every regular expression $\mathcal{R}$ without the union operator which describes an infinite language and which is at construction level of at most $n$, the language $L(\mathcal{R})$ can be represented as $L(\mathcal{R}) = L_lL_i^*L_r$ with $|L_l| < \infty$ and $L_i \notin \{\emptyset, \{\lambda\}\}$.
  Now, let $\cR$ be a regular expression of construction level $n+1$ which describes an infinite language and which does not contain the union operator. Then, there are two possibilities how $\cR$ is built: by concatenation of two regular expressions where for at least one of the described languages the induction hypothesis holds or by Kleene closure of a regular expression which neither describes the empty set nor the language $\{\lambda\}$ (otherwise, $L(\cR)$ would be finite).\smallskip
    
  \noindent{\sl Case 1}: Let $\mathcal{R} = \mathcal{S}\mathcal{T}$. Then, the equation $L(\mathcal{R}) = L(\mathcal{S})L(\mathcal{T})$ holds. If $L(\cS)$ is infinite, we get, according to the induction hypothesis, $L(\mathcal{S}) = S_lS_i^*S_r$ for suitable sets $S_l$, $S_i$, and $S_r$. With $R_l = S_l$, $R_i = S_i$, and~$R_r = S_rL(\cT)$, we obtain $L(\mathcal{R}) = R_lR_i^*R_r$ with $|R_l| < \infty$ and $R_i \notin \{\emptyset, \{\lambda\}\}$. If $L(\cS)$ is finite, then~$L(\cT)$ is infinite (because we consider only such $\cR$ where $L(\cR)$ is infinite) and we get, according to the induction hypothesis, that~$L(\mathcal{T}) = T_lT_i^*T_r$ for suitable sets $T_l$, $T_i$, and $T_r$. With $R_l = L(\cS)T_l$, $R_i = T_i$, and~$R_r = T_r$, we obtain a desired representation $L(\mathcal{R}) = R_lR_i^*R_r$ with $|R_l| < \infty$ and $R_i \notin \{\emptyset, \{\lambda\}\}$.\smallskip

  \noindent{\sl Case 2}: Let $\mathcal{R} = \mathcal{S}^*$. Then, the equation $L(\mathcal{R}) = (L(\mathcal{S}))^*$ holds. Thus, with $R_l = \{\lambda\}$, $R_i = L(\mathcal{S})$, and~$R_r = \{\lambda\}$, we obtain that~$L(\mathcal{R}) = R_lR_i^*R_r$ with~$|R_l| < \infty$ and $R_i \notin \{\emptyset, \{\lambda\}\}$.\smallskip

  Hence, every infinite union-free language can be expressed in the claimed form.
  \end{proof}
  
We proved with Lemma~\ref{lemma:2com_zerlegung} that any $\TCOM$-language can be given as a union of finitely many 
$\TCOM$-languages where the first comet tail is always union-free. Together, we obtain that any $\TCOM$-language has a
representation in the $\TCOM$-form where the first comet tail is a finite set.
  
\begin{lemma}\label{lemma:2com_zerlegung_e_uf_e_endlich}
For each two-sided comet language $L = EG^*H$ with $E \in \UF$, there exist a finite language $E'$, 
a language $G'\notin \{\emptyset, \{\lambda\}\}$, and a regular language $H'$ 
such that $L = E'(G')^*H'$.
\end{lemma}
\begin{proof}
We have shown in Lemma \ref{lemma:2com_unendlich_oder_leereSprache} that each two-sided comet language $L$ is either empty or
infinite. For the first case, the assertion holds with $E' = \emptyset$ and any regular 
languages $G' \notin \{\emptyset, \{\lambda\}\}$ and $H'$. 
     
Now, let $L=EG^*H$ be an infinite $\TCOM$-language with $E \in \UF$. If $E$ is finite, then we already have a desired form with 
$E'=E$, $G'=G$, and $H'=H$. 

So, let $E$ be infinite.
By Lemma~\ref{lemma:reg_aufteilung}, we know that there are languages $E_l$, $E_i$, and $E_r$ such 
that $E_l$ is a finite set, $E_i\notin\{\emptyset,\{\lambda\}\}$, and $E=E_lE_i^*E_r$.
If we set $E' = E_l$, $G' = E_i$, and $H' = E_rG^*H$, then we obtain a desired form because $L=E'(G')^*H'$ 
where $E'$ is finite, $G'\notin\{\emptyset,\{\lambda\}\}$, and $H'$ is a regular language.
\end{proof}
  
Now we connect the previous lemmas and conclude that, for every two-sided comet language, there is such a 
representation where the first comet tail of the language is finite.
  
\begin{theorem}[Normal form for $\TCOM$-languages]\label{theorem:linksendliche_nf_2com}
For each two-sided comet language, there exists a representation~$L = EG^*H$ such that $E$ is a finite language
and $G \notin \{\emptyset, \{\lambda\}\}$.
\end{theorem}
\begin{proof}
According to Lemma \ref{lemma:2com_zerlegung}, any two-sided comet language $L=E'(G')^*H'$ can be 
represented as a union of finitely many languages $E'_i(G')^*H'$ such that all languages~$E'_i$ are union-free. According to 
Lemma \ref{lemma:2com_zerlegung_e_uf_e_endlich}, every such language $E'_i(G')^*H'$ can in turn be represented 
as a $\TCOM$-language $E_iG^*H$ where the first tail $E_i$ is finite. The union $E$ of all these finite languages $E_i$
is also finite. Hence, we obtain 
\[L=E'(G')^*H'=\bigcup_{i=1}^{n} E'_i(G')^*H' = \bigcup_{i=1}^{n} E_iG^*H = \left(\bigcup_{i=1}^n E_i\right) G^*H =EG^*H
\] 
where $E$ is finite and $G \notin \{\emptyset, \{\lambda\}\}$.
\end{proof}
  
We refer to this representation as a left-sided normal form. A right-sided normal form (where the last comet tail is a 
finite set) can be derived in a similar way. 

\subsection{Hierarchy of subregular language classes}\label{sec:subreg}
  
  In this section, we investigate inclusion relations between various subregular languages classes.
  Figure~\ref{fig:lang_erg_1} shows the results. 

  \begin{figure}[htb]
  \centering
  \scalebox{.9}{\begin{tikzpicture}[node distance=15mm and 16mm, on grid]
    \node (MON) {$\MON$};
    \node (d0)[above=of MON] {};
    \node (FIN)[right=of d0] {$\FIN$};
    \node (NIL)[above=of d0] {$\NIL$};
    \node (COMB)[left=of NIL] {$\COMB$};
    \node (DEF)[above=of NIL] {$\DEF$};
    \node (d1)[left=of DEF] {};
    \node (SYDEF)[left=of d1] {$\SYDEF$};
    \node (d2)[right=of DEF] {};
    \node (SUF)[right=of d2] {$\SUF$};
    \node (ORD)[above=of DEF] {$\ORD$};
    \node (NC)[above=of ORD] {$\NC\stackrel{\text{\cite{McNaughton_Papert.1971}}}{=}\SF$};
    \node (PS)[above=of NC] {$\PS$};
    \node (RCOM)[above=of SYDEF] {$\RCOM$};
    \node (LCOM)[left=of RCOM] {$\LCOM$};
    \node (TCOM)[above=of RCOM] {$\TCOM$};
    \node (COMM)[right=of SUF] {$\COMM$};
    \node (CIRC)[above=of COMM] {$\CIRC$};
    \node (UF)[left=of LCOM] {$\UF$};
    \node (STAR)[below=of UF] {$\STAR$};
    \node (REG) [above of = PS] {$\REG$};
  
    \draw[hier, bend left] (MON) to node[edgeLabel] {\small\ref{lemma:mon_subset_star_subset_uf}} (STAR);
    \draw[hier, bend left] (MON) to node[edgeLabel] {\small\ref{lemma:mon_subset_sydef_subset_com_subset_2com}} (SYDEF);
    \draw[hier] (RCOM) to node[pos=.45,edgeLabel]{\small\cite{Bordihn_Holzer_Kutrib.2009}}(TCOM);
    \draw[hier] (LCOM) to node[edgeLabel]{\small\ref{lemma:mon_subset_sydef_subset_com_subset_2com}}(TCOM);
    \draw[hier, bend left=20] (TCOM) to node[edgeLabel]{\small\cite{Bordihn_Holzer_Kutrib.2009}}(REG);
    \draw[hier] (STAR) to node[pos=.45,edgeLabel]{\small\ref{lemma:mon_subset_star_subset_uf}} (UF);
    \draw[hier, bend left] (UF) to node[edgeLabel]{\small\cite{Holzer_Truthe.2015}}(REG);
    \draw[hier, bend right] (MON) to node[pos=.7,edgeLabel]{\small\cite{Truthe.2018}} (COMM);
    \draw[hier] (COMM) to node[pos=.45,edgeLabel]{\small\cite{Holzer_Truthe.2015}}(CIRC);
    \draw[hier, bend right] (CIRC) to node[edgeLabel]{\small{\cite{Holzer_Truthe.2015}}}(REG);
    \draw[hier] (MON) to node[pos=.45,edgeLabel]{\small\cite{Truthe.2018}} (NIL);
    \draw[hier] (NIL) to node[pos=.45,edgeLabel]{\small\cite{Wiedemann.1978}}(DEF);
    \draw[hier] (DEF) to node[pos=.45,edgeLabel]{\small\cite{Holzer_Truthe.2015}}(ORD);
    \draw[hier] (ORD) to node[pos=.45,edgeLabel]{\small\cite{Shyr_Thierrin.1974.ord}}(NC);
    \draw[hier] (NC) to node[pos=.4,edgeLabel]{\small\cite{Shyr_Thierrin.1974.ps}}(PS);
    \draw[hier] (PS) to node[pos=.45,edgeLabel]{\small\cite{Holzer_Truthe.2015}}(REG);
    \draw[hier, bend right] (MON) to node[pos=.7,edgeLabel]{\small\cite{Truthe.2018}}(SUF);
    \draw[hier, bend right] (SUF) to node[edgeLabel]{\small\cite{Holzer_Truthe.2015}}(PS);
    \draw[hier] (FIN) to node[pos=.45,edgeLabel]{\small\cite{Wiedemann.1978}}(NIL);
    \draw[hier] (COMB) to node[pos=.45,edgeLabel]{\small\cite{Havel.1969}}(DEF);
    \draw[hier] (COMB) to node[pos=.45,edgeLabel]{\small\cite{Olejar_Szabari.2023}}(SYDEF);
    \draw[hier] (SYDEF) to node[pos=.45,edgeLabel]{\small\ref{lemma:mon_subset_sydef_subset_com_subset_2com}}(LCOM);
    \draw[hier] (SYDEF) to node[pos=.45,edgeLabel]{\small\cite{Olejar_Szabari.2023}}(RCOM);
    \draw[hier] (SYDEF) to node[edgeLabel]{\small\cite{Olejar_Szabari.2023}}(PS);
  \end{tikzpicture}}
  \caption{Resulting hierarchy of subregular language families}
  \label{fig:lang_erg_1}
  \end{figure}

  An arrow from a node $X$ to a node~$Y$ stands for the proper inclusion $X \subset Y$. 
  If two families are not connected by a directed path, then they are incomparable. 
  An edge label refers to the paper where the proper inclusion has been shown (in some cases, it might be that it is not the 
  first paper where the respective inclusion has been mentioned, since it is so obvious that it was not emphasized in a 
  publication) or the lemma of this paper where the proper inclusion will be shown.

  In the literature, it is often said that two languages are equivalent if they are equal or differ
  at most in the empty word. Similarly, two families can be regarded to be equivalent if they differ 
  only in the languages $\emptyset$ or $\{\lambda\}$. Therefore, the set $\STAR$ of all star languages is sometimes regarded as a proper subset of the set $\COM$ of all (left-, right-, or two-sided) comet languages although $\{\lambda\}$ belongs to the family $\STAR$ but not to $\LCOM$, $\RCOM$ 
  or $\TCOM$. We regard $\STAR$ and $\STAR\setminus\{\{\lambda\}\}$ as different. Then, the family
  $\STAR$ is incomparable to $\LCOM$, $\RCOM$, and $\TCOM$, as we will later show.

For space reasons, we give the following observation without a proof.

  \begin{lemma}\label{lemma:lcom_rcom}
  Whenever a language $L$ is a right-sided comet then its reversal $L^R$ is a left-sided comet language and vice versa.
  \end{lemma}
 
  \begin{corollary}\label{cor:lcom_rcom}
  We have $\LCOM=\set{L^R}{L\in\RCOM}$ and $\RCOM=\set{L^R}{L\in\LCOM}$.
  \end{corollary}
  
  We now present some languages which will serve later as witness languages for proper inclusions or
  incomparabilities.
  
  
  \begin{lemma}\label{lemma:star_o_2com}
  The language $L = \{\lambda\}$ is in $\STAR\setminus \TCOM.$
  \end{lemma}
  \begin{proof}
  The language $L$ is a star language since $L = H^*$  with $H = \{\lambda\}$. 
  According to Lemma \ref{lemma:2com_unendlich_oder_leereSprache}, a two-sided comet language is either 
  infinite or the empty language. Hence, $L$ is not a two-sided comet.
  \end{proof}
   
  \begin{lemma}\label{lemma:star_com_o_ps}
  Let $L = \set{a^{2n}}{n \geq 0}$. Then, it holds $L\in(\STAR\cap\LCOM\cap\RCOM)\setminus \PS$.
  \end{lemma}
  \begin{proof}
  Let $G = \{aa\}$ and $E=H=\{\lambda\}$.
  The language $L$ can be expressed as $L = G^* = EG^* = G^*H$. Therefore, $L\in\STAR\cap\LCOM\cap\RCOM$.
  
  Assume that $L\in\PS$. Then, there is a natural number $m\geq 1$ such that, for any word~$x\in \{a\}^*$, 
  either~$J_x^m \cap L = \emptyset$ or~$J_x^m\subseteq L$ where $J_x^m = \set{ x^n}{n\geq m}$.
  For any natural number $m\geq 1$, we have with the word~$x=a$ the set $J_a^m = \set{ a^n}{n\geq m}$. Since
  $a^{2m}\in J_a^m \cap L$, the intersection is not empty. But, since~$a^{2m+1}\in J_a^m\setminus L$, it neither
  holds $J_a^m\subseteq L$. Hence, the language $L$ is not power-separating.
  \end{proof}
  
  \begin{lemma}\label{lemma:star_com_o_circ}
  Let $L = \{ab\}^*$. Then, it holds $L \in (\STAR\cap\LCOM\cap\RCOM)\setminus \CIRC.$
  \end{lemma}
  \begin{proof}
  Let $G = \{ab\}$ and $E=H=\{\lambda\}$.
  The language $L$ can be expressed as $L = G^* = EG^* = G^*H$. Therefore, $L \in \STAR\cap\LCOM\cap\RCOM$.
  
  Assume that the language $L$ is circular. Then, the word $ba$ would belong to it because $ab\in L$ but it does not.
  Hence, $L\notin \CIRC$.
  \end{proof}
  
  \begin{lemma}[\cite{Olejar_Szabari.2023}]\label{lemma:sydef_o_sf}
  Let $V=\{a,b\}$ be an alphabet, $H=\{ba\}\{b\}^*(\{aa\}\{b\}^*)^*$ a regular language over~$V$, and $L=V^*H$. 
  Then, $L\in\SYDEF\setminus\SF$.
  \end{lemma}
  \begin{proof}
  The language $L$ can be represented as $\{\lambda\}V^*H$. So, the language is symmetric definite.
  As shown in \cite{Olejar_Szabari.2023}, the language is not star-free.
  \end{proof}
  
  
  \begin{lemma}\label{lemma:lcom_o_rcom}
  Let $L_1=\set{a^nb}{n \geq 0}$ and $L_2=L_1^R$. Then, $L_1\in\RCOM\setminus \LCOM$
  and $L_2\in\LCOM \setminus \RCOM$.
  \end{lemma}
  \begin{proof}
  The language $L_1$ can be expressed as $\{a\}^*\{b\}$, hence, in the form $L_1=G^*H$ with $G=\{a\}$ and~$H=\{b\}$.
  Thus, $L_1\in\RCOM$.
  
  Assume that $L_1\in\LCOM$. Then, two languages $E$ and $I$ would exist such that $L_1 = EI^*$.
  Since $b$ is a suffix of every word in $L_1$, the letter $b$ is also a suffix of a word in $I$.
  But then $L_1$ would also contain a word with more than one $b$ which is a contradiction.
  Hence, $L_1\notin\LCOM$.
  
  By Corollary~\ref{cor:lcom_rcom}, it follows that $L_2\in\LCOM\setminus\RCOM$.
  \end{proof}
  
  \begin{lemma}\label{lemma:fin_suf_comm_o_star_2com}
  The language $L = \{\lambda,a\}$ belongs to the set $(\FIN\cap\SUF\cap\COMM) \setminus (\STAR\cup\TCOM)$.
  \end{lemma}
  \begin{proof}
  All suffixes of all words of the language $L$ belong to $L$. Thus, $L$ is suffix-closed. Furthermore, the language is finite
  but not empty and commutative.
  According to Lemma~\ref{lemma:2com_unendlich_oder_leereSprache}, each two-sided comet language is either empty or infinite. 
  Hence, $L$ is not a two-sided comet language. According to Lemma~\ref{lemma:star_unendlich_oder_leerwort},
  each star language is either infinite or contains only the empty word. Hence, $L$ is not a star language either.
  \end{proof}

  We now prove some proper inclusions.

  \begin{lemma}\label{lemma:mon_subset_star_subset_uf}
  We have the proper inclusions $\MON \subset \STAR\subset \UF$.
  \end{lemma}
  \begin{proof}
  We first prove the relation $\MON \subset \STAR$: Any monoidal language can be expressed as $L = V^*$ for some alphabet $V$. 
  Since $V$ is a regular language, $L$ is a star language. A witness language for the properness is the 
  language $L = \set{a^{2n}}{n \geq 0}$ as shown in Lemma~\ref{lemma:star_com_o_ps}.
  
  We now prove the relation $\STAR\subset \UF$: Every language $H^*$ for some regular language $H$ is union-free
  according to~\cite{Nagy.2019}.
  A witness language for the properness is $L = \{a\}$ which is union-free but,
  according to Lemma~\ref{lemma:star_unendlich_oder_leerwort}, not a star language since 
  it is neither infinite nor equal to $\{\lambda\}$.
  \end{proof}

  \begin{lemma}\label{lemma:mon_subset_sydef_subset_com_subset_2com}
  We have the proper inclusions $\MON\subset\SYDEF\subset \cC\subset\TCOM$ for $\cC\in\{\LCOM,\RCOM\}$.
  \end{lemma}
  \begin{proof*}
  \begin{enumerate}
  \item $\MON\subset\SYDEF$: Any monoidal language can be expressed as $L = V^*$ for some alphabet $V$ and, with $E=H=\{\lambda\}$
  also in the form $EV^*H$. Hence, the language $L$ is symmetric definite. A witness language for the properness 
  is $\{a,b\}^*\{ba\}\{b\}^*(\{aa\}\{b\}^*)^*$ from Lemma~\ref{lemma:sydef_o_sf} 
  (and originally~\cite{Olejar_Szabari.2023}).
  
  \item $\SYDEF\subset \RCOM$: This relation was proved in \cite{Olejar_Szabari.2023}.
  
  \item $\SYDEF\subset \LCOM$: The family $\SYDEF$ is closed under reversal. For any symmetric definite language $L$,
  its reversal $L^R$ also belongs to the family $\SYDEF$ and, by \cite{Olejar_Szabari.2023}, is also a right-sided
  comet language. By Lemma~\ref{lemma:lcom_rcom}, the reversal of the language $L^R$, hence $L$ itself, is a left-sided 
  comet language. A witness language for the properness is the language $L = \set{a^{2n}}{n \geq 0}$
  according to Lemma~\ref{lemma:star_com_o_ps} where it is shown that $L\in\LCOM\setminus \PS$
  and according to \cite{Olejar_Szabari.2023} where the inclusion~$\SYDEF\subset\PS$ is proved.
  
  \item $\RCOM\subset \TCOM$: This relation was proved in \cite{Olejar_Szabari.2023}.
  
  \item $\LCOM\subset \TCOM$: Any left-sided comet language $L=EG^*$ is also a two-sided comet $EG^*H$ with~$H=\{\lambda\}$.
  In Lemma~\ref{lemma:lcom_o_rcom}, it was shown that the language $L=\set{a^nb}{n \geq 0}$ is a right-sided comet
  language but not a left-sided comet. By \cite{Olejar_Szabari.2023}, it is a two-sided comet language.\hspace*{\fill}$\Box$
  \end{enumerate}
  \end{proof*}

  We now prove the incomparability relations mentioned in Figure~\ref{fig:lang_erg_1} which have not been proved
  earlier. These are the relations regarding the families $\STAR$, $\SYDEF$, $\LCOM$, $\RCOM$, and $\TCOM$.

\begin{lemma}\label{lemma:star_uf_unvergleichbarzu_comb_2com}
  Each of the families $\STAR$ and $\UF$ is incomparable to each of the 
  families $\COMB$, $\SYDEF$, $\RCOM$, $\LCOM$, and $\TCOM$.
\end{lemma}
\begin{proof}
  Due to inclusion relations, it suffices to show that there are a language $L_1\in \STAR\setminus\TCOM$ 
  and a language $L_2\in\COMB\setminus\UF$.
  From Lemma~\ref{lemma:star_o_2com}, we get $L_1=\{\lambda\}$.
  From \cite{Holzer_Truthe.2015}, we take $L_2=\{a,b,c\}^*\{a,b\}$.
\end{proof}
  
\begin{lemma}\label{lemma:star_unvergleichbarzu_fin_ps}
  The language family $\STAR$ is incomparable to each of the 
  families $\FIN$, $\NIL$, $\DEF$, $\ORD$, $\NC$, $\SF$, $\PS$, and $\SUF$.
\end{lemma}
\begin{proof}
  Due to inclusion relations, it suffices to show that there are a language $L_1\in \STAR\setminus\PS$, 
  a language~$L_2\in\FIN\setminus\STAR$, and a language $L_3\in\SUF\setminus\STAR$.
  As $L_1$, we obtain from Lemma~\ref{lemma:star_com_o_ps} the language~$L_1 = \set{a^{2n}}{n \geq 0}$.
  From Lemma~\ref{lemma:fin_suf_comm_o_star_2com}, we take $L_2 = L_3 = \{\lambda,a\}$.
\end{proof}

\begin{lemma}\label{lemma:star_unvergleichbarzu_comm_circ}
  The language family $\STAR$ is incomparable to the families $\CIRC$ and $\COMM$.
\end{lemma}
\begin{proof}
  Due to inclusion relations, it suffices to show that there are a language $L_1\in \STAR\setminus\CIRC$ 
  and a language $L_2\in\COMM\setminus\STAR$.
  From Lemma~\ref{lemma:star_com_o_circ}, we have $L_1 = \{ab\}^*$.
  From Lemma~\ref{lemma:fin_suf_comm_o_star_2com}, we take again the language $L_2 = \{\lambda,a\}$.
\end{proof}

\begin{lemma}\label{lemma:lcom_unvergleichbarzu_rcom}
  The language families $\LCOM$ and $\RCOM$ are incomparable to each other.
\end{lemma}
\begin{proof}
  With the witness languages
  $L_1=\set{a^nb}{n \geq 0}\in\RCOM\setminus \LCOM$ and $L_2=L_1^R\in\LCOM \setminus \RCOM$,
  the statement follows from Lemma~\ref{lemma:lcom_o_rcom}.
\end{proof}

\begin{lemma}\label{lemma:sydef_2com_unvergleichbarzu_fin_nc}
  The language families $\SYDEF$, $\LCOM$, $\RCOM$, and $\TCOM$ are incomparable to each of the families $\FIN$,
  $\NIL$, $\DEF$, $\ORD$, $\NC$, and $\SF$.
\end{lemma}
\begin{proof}
  Due to inclusion relations, it suffices to show that there are a language $L_1\in \SYDEF\setminus\SF$ 
  and a language $L_2\in\FIN\setminus\TCOM$.
  From Lemma~\ref{lemma:sydef_o_sf} (and previously \cite{Olejar_Szabari.2023}), for the first language,
  we obtain the language~$L_1 = \{a,b\}^*\{ba\}\{b\}^*(\{aa\}\{b\}^*)^*$.
  From Lemma~\ref{lemma:fin_suf_comm_o_star_2com}, we take the language $L_2 = \{\lambda,a\}$.
\end{proof}

\begin{lemma}\label{lemma:com_2com_unvergleichbarzu_ps}
  The language families $\LCOM$, $\RCOM$, and $\TCOM$ are incomparable to the family~$\PS$.
\end{lemma}
\begin{proof}
  Due to inclusion relations, it suffices to show that there are a language $L_1\in \LCOM\setminus\PS$ 
  and a language $L_2\in\PS\setminus\TCOM$. The property of (non) power-separating is not influenced 
  by the reversal operation. If there is a language $L_1\in \LCOM\setminus\PS$, then there is
  also a language in the set $\RCOM\setminus\PS$, namely~$L_1^R$.
  From Lemma~\ref{lemma:star_com_o_ps}, we have $L_1 = \set{a^{2n}}{n \geq 0}\in (\LCOM\cap\RCOM)\setminus\PS$.
  As language $L_2$, we take again the language $L_2 = \{\lambda,a\}$ from Lemma~\ref{lemma:fin_suf_comm_o_star_2com}.
\end{proof}

\begin{lemma}\label{lemma:sydef_2com_unvergleichbarzu_suf_comm_circ}
  The language families $\SYDEF$, $\LCOM$, $\RCOM$, and $\TCOM$ are incomparable to each of the 
  families $\SUF$, $\CIRC$ and $\COMM$.
\end{lemma}
\begin{proof}
  Due to inclusion relations, it suffices to show that there are a language $L_1\in \SYDEF\setminus\SUF$,
  a language $L_2\in \SYDEF\setminus\CIRC$, a language $L_3\in\SUF\setminus\TCOM$, and
  a language $L_4\in\COMM\setminus\TCOM$.
  In~\cite{Holzer_Truthe.2015}, it was shown that the families $\COMB$ and $\SUF$ are disjoint. 
  Since~$\COMB\subseteq\SYDEF$, we can take any combinational language as $L_1$, for instance, $L_1=\{a,b\}^*\{b\}$.
  The same language serves as $L_2$ because it is not circular.
  From Lemma~\ref{lemma:fin_suf_comm_o_star_2com}, we take again the language $\{\lambda,a\}$ as $L_3$ and $L_4$.
\end{proof}

From all these relations, the hierachy presented in Figure~\ref{fig:lang_erg_1} follows.

\begin{theorem}[Resulting hierarchy]\label{theorem:neue_hierarchie}
The inclusion relations presented in Figure~\ref{fig:lang_erg_1} hold. An arrow from an entry $X$ to
an entry~$Y$ depicts the proper inclusion $X \subset Y$; if two families are not connected by a directed
path, then they are incomparable.
\end{theorem}
\begin{proof}
An edge label refers to the paper or lemma in the present paper where the proper inclusion is shown.
The incomparability results are proved in Lemmas~\ref{lemma:star_uf_unvergleichbarzu_comb_2com} 
to~\ref{lemma:sydef_2com_unvergleichbarzu_suf_comm_circ}.
\end{proof}

\section{Results on subregular control in external contextual grammars}

In this section, we include the families of languages generated by external contextual grammars with 
selection languages from the subregular families under investigation into the existing hierarchy 
with respect to external contextual grammars.

If, in a contextual grammar, all selection languages belong to some language family $X$, then they belong also
to every super set $Y$ of $X$. Therefore, each language in $\ec{X}$ is also generated by a contextual grammar
with selection languages from $Y$ and we have the following monotonicity.

\begin{lemma}\label{lemma:ec_monoton}
  For any two language classes $X$ and $Y$ with $X\subseteq Y$,
  we have the inclusion
  $\cEC(X)\subseteq\cEC(Y)$.
\end{lemma}

Figure~\ref{fig:lang_erg_2} shows a hierarchy of some language families which are generated by external
contextual grammars where the selection languages belong to subregular classes investigated before. The 
hierarchy contains results which were already known (marked by a reference to the literature) and 
results which will be proved in this section (marked by a number which refers to the respective lemma). 

  \begin{figure}[htb]
  \centering
  \scalebox{.9}{
  \begin{tikzpicture}[node distance=15mm and 25mm,on grid=true
  ]
    \node (MON) {$\ec{\MON}$};
    \node (FIN) [below=of MON] {$\ec{\FIN}$};
    \node (COMB)[above=of MON] {$\ec{\COMB}$};
    \node (NIL) [right=of COMB] (NIL) {$\ec{\NIL}$};
    \node (DEF) [above=of COMB] {$\ec{\DEF}$};
    \node (ORD) [above=of DEF] {$\ec{\ORD}$};
    \node (SYDEF) [left=of ORD] {$\ec{\SYDEF}$};
    \node (SUF) [below left=of SYDEF] {$\ec{\SUF}$};
    \node (COMM) [right=of ORD] {$\ec{\COMM}$};
    \node (STAR) [below right=of COMM] {$\ec{\STAR}$};
    \node (NC) [above=of ORD] {$\ec{\NC}$};
    \node (PS) [above=of NC] {$\ec{\PS}$};
    \node (CIRC) [above=of COMM] {$\ec{\CIRC}$};
    \node (REG) [above=of PS] {$\ec{\REG} \stackrel{\text{\cite{Dassow_Manea_Truthe.2012}}}{=} \ec{\UF} \stackrel{\text{\ref{lemma:ec:reg_eq_com}}}{=} \cEC(\LCOM) \stackrel{\text{\ref{lemma:ec:reg_eq_com}}}{=} \cEC(\RCOM) \stackrel{\text{\ref{lemma:ec:reg_eq_com}}}{=} \cEC(\TCOM)$};

    \draw[hier] (FIN) to node[pos=.45, edgeLabel]{\small\cite{Dassow.2005}} (MON);
    \draw[hier, bend right] (MON) to node[pos=.45, edgeLabel]{\small\cite{Dassow.2005}} (NIL);
    \draw[hier] (MON) to node[pos=.45, edgeLabel]{\small\cite{Dassow.2015}} (COMB);
    \draw[hier] (COMB) to node[pos=.45, edgeLabel]{\small\cite{Truthe.2021}} (DEF);
    \draw[hier] (DEF) to node[pos=.45, edgeLabel]{\small\cite{Truthe.2014}} (ORD);
    \draw[hier, bend left] (DEF) to node[pos=.45, edgeLabel]{\small\ref{lemma:ec:def_ss_sydef}} (SYDEF);
    \draw[hier] (ORD) to node[pos=.45, edgeLabel]{\small\cite{Dassow_Truthe.2023}} (NC);
    \draw[hier] (NC) to node[pos=.45, edgeLabel]{\small\cite{Truthe.2021}} (PS);
    \draw[hier, bend left] (MON) to node[pos=.45, edgeLabel]{\small\cite{Dassow.2005}} (SUF);
    \draw[hier, bend left=35] (SUF) to node[pos=.45, edgeLabel]{\small\cite{Truthe.2021}} (PS);
    \draw[hier, bend right=35] (MON) to node[pos=.45, edgeLabel]{\small\ref{lemma:ec:mon_ss_star}} (STAR);
    \draw[hier, bend right=35] (STAR) to node[pos=.45, edgeLabel]{\small\ref{lemma:ec:star_ss_rcom}} (REG);
    \draw[hier, bend left] (SYDEF) to node[pos=.45, edgeLabel]{\small\ref{lemma:ec:sydef_ss_ps}} (PS);
    \draw[hier] (NIL) to node[pos=.45, edgeLabel]{\small\cite{Dassow.2005}} (COMM);
    \draw[hier, bend right] (NIL) to node[pos=.45, edgeLabel]{\small \cite{Dassow.2005}} (DEF);
    \draw[hier] (COMM) to node[pos=.45, edgeLabel]{\small\cite{Dassow_Manea_Truthe.2012}} (CIRC);
    \draw[hier, bend right] (CIRC) to node[pos=.45, edgeLabel]{\small\cite{Dassow_Manea_Truthe.2012}} (REG);
    \draw[hier] (PS) to node[pos=.45, edgeLabel]{\small\cite{Truthe.2021}} (REG);
  \end{tikzpicture}
  }
  \caption{Resulting hierarchy of language families by external contextual grammars with special selection languages}
  \label{fig:lang_erg_2}
  \end{figure}

  An arrow from a node $X$ to a node~$Y$ stands for the proper inclusion $X \subset Y$. 
  If two families are not connected by a directed path, then they are incomparable. 
  An edge label refers to the paper where the proper inclusion has been shown 
  or the lemma of this paper where the proper inclusion will be shown.

We now present some languages which will serve later as witness languages for proper inclusions or
incomparabilities. Due to space limitations, we give only proof sketches in some cases where we believe
that the reader finds the idea feasible.

\begin{lemma}\label{lemma:ec:nil_o_star}
  Let $L = \set{a^nbbb}{n\geq 1} \cup \{\lambda\}$. Then, it holds $L \in \ec{\NIL} \setminus \ec{\STAR}$.
\end{lemma}
\begin{proof}
  The contextual grammar 
  $G = (\{a,b\}, \sets{\{a,b\}^*\{a,b\}^4\to(a,\lambda)}, \{abbb, \lambda\})$
  generates 
  $L$.

During the derivation, the number of the letter $a$ is increasing without changing the number of $b$. If the selection languages
are from $\STAR$, then such a context containing letters $a$ only could be wrapped around the empty word yielding a word 
without $b$ which is a contradiction.
\end{proof}

\begin{lemma}\label{lemma:ec:comb_o_pre_star}
  Let $L = \set{b^na}{n \geq 0}\cup \{\lambda\}$. Then, it holds $L \in \ec{\COMB}\setminus \ec{\STAR}$.
\end{lemma}
\begin{proof}
  The contextual grammar $G = (\{a,b\}, \sets{\{a,b\}^*\{a\}\to(b, \lambda)}, \{\lambda,a\})$
  generates the language~$L$ and the selection language is combinational. 

  Similarly to the proof before: With star selection languages, a word with the letter $b$ but without $a$ could be generated.
\end{proof}

\begin{lemma}\label{lemma:ec:suf_o_star}
Let $L_1=\sets{a,b}^*\set{a^n b^m}{n\geq 1,\ m\geq 1}$, $L_2=\set{ca^n b^mc}{n\geq 1,\ m\geq 1}$, and $L=L_1\cup L_2$.
Then, it holds $L\in\ec{\SUF} \setminus \ec{\STAR}$.
\end{lemma}
\begin{proof}
  It holds $L=L(G)$ for the contextual grammar $G = (\{a,b,c\}, \sets{(S_1,C_1),(S_2,C_2)}, \{ab\})$ with
  \[S_1=\sets{a,b}^*,\quad
    C_1=\sets{(a, \lambda),(b, \lambda),(\lambda, b)},\qquad
    S_2=\set{a^n b^m}{n\geq 0,\ m\geq 1 }\cup\sets{\lambda},\quad
    C_2=\sets{(c,c)}.
  \]

Using star selection languages, the two letters $c$ could be wrapped around a word with more than 
one~$a$-to-$b$-change from $L_1$
which would yield a word not belonging to $L$.
%
\end{proof}

\begin{lemma}\label{lemma:ec:star_o_ps}
  Let $L_1 = \set{a^n}{n \geq 2}$ and $L_2 = \set{ba^{2n}b}{n \geq 1}$ be two languages and $L = L_1 \cup L_2$ its union.
  Then, the relation~$L\in\ec{STAR} \setminus \ec{PS}$ holds.
\end{lemma}
\begin{proof}
  It holds $L=L(G)$ for the contextual grammar 
  $G = (\{a,b\}, \sets{(S_1,C_1),(S_2,C_2)}, \{aa\})$ with
  \[
    S_1 = \set{a^{n}}{n \geq 0},\quad C_1 = \sets{(\lambda, a)} \quad\mbox{and}\quad
    S_2 = \set{a^{2n}}{n \geq 0},\quad C_2 = \sets{(b, b)}.
  \]
%

  Now assume that $L \in \ec{\PS}$. Then, $L = L(G')$ for a contextual 
  grammar $G'$
  where every selection language 
  is power-separating.
  
  For every selection language (since it is power-separating), there is a number $m_S\in \N$ such that,
  for every word $x\in \{a,b\}^*$, either~$J_x^{m_S} \cap S = \emptyset$ or $J_x^{m_S} \subseteq S$ 
  with $J_x^{m_S} = \set{x^n}{n \geq m_S}$. Let $m_S$ be the minimum of these numbers for $S$ and let $m$ be the maximum
  of all the values $m_S$ for a selection language $S$.
  
  Further, let $p = m+\ell(G')$. Then, we have the following statement for every selection language~$S$:
  For each word $x \in \{a,b\}^*$, it is 
  \begin{equation}\label{eq:ps}
    \text{either }J_x^p \cap S = \emptyset \text{ or } J_x^p \subseteq S
  \end{equation}
  where $J_x^p = \set{x^n}{n \geq p}$.

  The language $L_2$ contains words with an arbitrary even number of letters $a$ and a letter $b$ at each end. 
  Hence, there is a derivation $w_0 \Lra^* w_1 \Lra uw_1v$
  with $w_0 \in A$, $|w_1|_a > p$, $|w_1|_b = 0$, and $|uv|_b > 0$. This implies $w_1 = a^k$ with $k > p$.

  Let $S$ be the selection language used in the last derivation step. Then, we have $a^k\in S$ and, with property~(\ref{eq:ps}),
  also $a^{k+1} \in S$. Since $a^{k+1}$ belongs to $L_1$ and therefore also to $L$, the last derivation step can also be 
  applied to $a^{k+1}$ which yields the word $ua^{k+1}v$. Since $|uv|_b > 0$, the word $ua^{k+1}v$ belongs at most to $L_2$.
  Since $ua^{k}v\in L_1$, we know that $|ua^kv|_a$ is an even number and~$|ua^{k+1}v|_a$ is an odd number. 
  Therefore, the word~$ua^{k+1}v$ does not belong to $L_2$ and neither to $L$ which is a contradiction to $L=L(G')$.
  Thus, we conclude~$L\notin\cEC(\PS)$.
\end{proof}

\begin{lemma}\label{lem:ec:star_o_circ}
  Let $L = \set{a^nb^n}{n \geq 1} \cup \set{b^na^n}{n \geq 1}$. Then, it holds $L\in\ec{\STAR} \setminus \ec{\CIRC}$.
\end{lemma}
\begin{proof}
  It holds $L=L(G)$ for the contextual grammar $G = (\{a,b\}, \sets{(S_1,C_1),(S_2,C_2)}, \{ab,ba\})$ with
  \[ S_1 = \set{a^n b^m}{n\geq 1,\ m\geq 1}^*,\quad C_1 = \sets{(a, b)}\quad\mbox{and}\quad
     S_2 = \set{b^n a^m}{n\geq 1,\ m\geq 1}^*,\quad C_2 = \sets{(b, a)}.
  \]

  With circular selection languages, a context $(a^k,b^k)$ could be wrapped around a word $b^ma^m$ yielding a word
  which does not belong to the language $L$.
%
\end{proof}

\begin{lemma}\label{lem:ec:ord_o_sydef}
  The language $L = \{a,b\}^*\cup \{c\}\{ab\}^*\{c\}$ belongs to the set $\ec{\ORD}\setminus \ec{\SYDEF}$.
\end{lemma}
\begin{proof}
  In Example~\ref{ex-cg}, we have given a contextual grammar where all selection languages are accepted by
  ordered finite automata, and thus, have shown that $L\in\cEC(\ORD)$.
  
  Suppose that the language $L$ is also generated by a contextual grammar $G'$
  where all selection languages 
  are symmetric definite.
  
  Let us consider a word $w=c(ab)^nc \in L$ for some $n \geq \ell(G')$. Due to the choice of $n$, the word~$w$
  is derived in one step from some word $z$ by using a selection language $S$ and context $(u,v)$: 
  $z\Lra uzv=w$. The word $u$ begins with the letter $c$; the word $v$ ends with $c$. Due to the choice of $n$, we also
  have~$|z|_a>0$ and $|z|_b>0$.
  Since $S$ is symmetric definite over the alphabet $V=\{a,b\}$, it can be expressed as $S=EV^*H$ for some 
  regular languages $E$ and $H$ over $V$. The sets $E$ and $H$ are not empty because $S$ contains at least the word $z$.
  Let $e$ be a word of $E$ and $h$ a word of $H$. Then, the word $ebbh$ belongs to the selection language $S$ as well.
  Since $ebbh\in\{a,b\}^*$ and $\{a,b\}^*\subseteq L$, we can apply the same derivation to this word and obtain $uebbhv$.
  This word starts and ends with $c$ but it does not have the form of those words from $L$ because of the double $b$.
  From this contradiction, it follows $L\notin\cEC(\SYDEF)$. 
\end{proof}

\begin{lemma}\label{lem:ec:suf_o_sydef}
  The language $L = \{a,b\}^*\cup \{c\}\{\lambda, b\}\{ab\}^*\{c\}$ belongs to $\ec{\SUF}\setminus \ec{\SYDEF}$.
\end{lemma}
\begin{proof}
  The language $L$ is generated by the contextual grammar $G=(\sets{a,b,c},\sets{(S_1,C_1),(S_2,C_2)},\sets{\lambda})$ with 
  \[S_1 = \sets{a,b}^*,\quad C_1 = \sets{(\lambda,a),(\lambda,b)}\quad\mbox{and}\quad
    S_2 = \Suf(\sets{ab}^*),\quad C_2 = \sets{(c,c)}\]
  where $\Suf(M)$ denotes the suffix-closure of the set $M$. 
  

  With the same argumentation as in the proof of Lemma~\ref{lem:ec:ord_o_sydef}, one can show also here $L\notin\cEC(\SYDEF)$
  (the letters $c$ are in both cases wrapped around words which are an alternating sequence of $a$ and $b$ what
  cannot be checked by a symmetric definite selection language).
\end{proof}

\begin{lemma}\label{lem:ec:sydef_o_nc}
  Let $L_1= \set{a^n}{n \geq 1}$, $L_2= \set{ba^nb}{n\geq 1}$, $L_3=  \set{cba^{2n}bc}{n \geq 1}$, 
  and $L = L_1 \cup L_2 \cup L_3$. Then, it holds $L\in \ec{\SYDEF}\setminus \ec{\NC}$.
\end{lemma}
\begin{proof}
  Let $V=\sets{a,b,c}$. The contextual grammar 
  $G=(V,\sets{(S_1,C_1),(S_2,C_2)},\sets{a})$ with 
  \[S_1 = \{a\}V^*\{\lambda\},\quad C_1 = \sets{(\lambda,a),(b,b)}\quad\mbox{and}\quad
    S_2 = \set{ba^{2m}b}{m\geq 1}V^*\{\lambda\},\quad C_2 = \sets{(c,c)}\]
  generates the language $L$. 
  This can be seen as follows: The shortest word of $L$ is $a$ which is the axiom. To every word of $L$ starting with the
  letter $a$ (hence, any word of $L_1$), another $a$ can be added or the letter~$b$ is added at the beginning 
  and the end of the word (using the first selection component) yielding all and only words of the languages $L_1$ 
  and $L_2$. To every word of $L_2$ which also belongs to $S_2$, the letter $c$ is added at the beginning and the end of the
  word (using the second selection component) yielding exactly the words of the language $L_3$. To the words of $L_3$,
  no selection component can be applied. All the selection languages are symmetric definite as can be seen from the form
  in which they are given.
  
  In \cite{Truthe.2021}, it was proved that the language $L$ does not belong to the family $\cEC(\NC)$.
\end{proof}

Next, we show some equalities.

\begin{lemma}\label{lemma:ec:reg_eq_com}
  A restriction to comet languages (left, right, two-sided) as selection languages does not decrease the generative capacity
  of external contextual grammars: 
  \[\ec{\REG} = \ec{\LCOM} = \ec{\RCOM} = \ec{\TCOM}.\]
\end{lemma}
\begin{proof}
  With the inclusions $\LCOM\subseteq \TCOM$, $\RCOM\subseteq \TCOM$, and $\TCOM\subseteq \REG$
  (see Theorem~\ref{theorem:neue_hierarchie} and Figure~\ref{fig:lang_erg_1}), we obtain also the inclusions
  $\cEC(\LCOM)\subseteq \cEC(\TCOM)$, $\cEC(\RCOM)\subseteq \cEC(\TCOM)$, and~$\cEC(\TCOM)\subseteq \cEC(\REG)$
  according to Lemma \ref{lemma:ec_monoton}.
  
  Let $G = (V, \{(S_1, C_1), \dots, (S_n,C_n)\}, A)$ be a contextual grammar 
  with arbitrary regular selection languages. Further, let $X$ be a new symbol ($X \notin V$). 
  We set $S_i' = \{X\}^*S_i$ for $1 \leq i \leq n$. Then, the contextual grammar 
  $G' = (V \cup \{X\}, \{(S_1', C_1), \dots, (S_n',C_n)\}, A)$ generates the same language as $G$. 
  The selection languages are all right-sided comet languages. The letter $X$ neither occurs in an axiom
  nor in a context. Therefore, the part $\{X\}^*$ of the selection languages has no impact on the possible
  derivations (the only word used is $\lambda$). Thus, the inclusion $\ec{\REG}\subseteq \ec{\RCOM}$ holds. 
  
  With $S_i' = S_i\{X\}^*$ for $1 \leq i \leq n$, the same language is generated and the selection languages are
  left-sided comets. Hence, we also have the inclusion $\ec{\REG}\subseteq \ec{\LCOM}$. 
  Hence, we obtain the chain of inclusions $\cEC(\REG)\subseteq \cC \subseteq \TCOM\subseteq \cEC(\REG)$ for $\cC\in\{\LCOM,\RCOM\}$ which implies the
  equalities stated in the lemma.
\end{proof}

We now prove some proper inclusions.

\begin{lemma}\label{lemma:ec:mon_ss_star}
  The family $\ec{\MON}$ is a proper subset of the family $\ec{\STAR}$.  
\end{lemma}
\begin{proof}
  With the inclusion $\MON\subseteq \STAR$ (see Theorem~\ref{theorem:neue_hierarchie} and Figure~\ref{fig:lang_erg_1}), 
  we obtain also the inclusion~$\cEC(\MON)\subseteq \cEC(\STAR)$
  according to Lemma \ref{lemma:ec_monoton}.
  
  The language $L=\set{a^n}{n \geq 2}\cup \set{ba^{2n}b}{n \geq 1}$ from Lemma~\ref{lemma:ec:star_o_ps}
  belongs to the family $\ec{\STAR}$ but not to the family $\ec{\PS}$ and, hence, neither to $\ec{\MON}$.
  Thus, the language is a witness for the properness of the inclusion.
\end{proof}

\begin{lemma}\label{lemma:ec:fin_ss_inf_pre_star}
  The family $\ec{\FIN}$ is a proper subset of the family $\ec{\STAR}$
\end{lemma}
\begin{proof}
According to \cite{Dassow.2005}, $\cEC(\FIN)\subset\cEC(\MON)$. According to Lemma~\ref{lemma:ec:mon_ss_star},
$\ec{\MON}\subset\ec{STAR}$. Hence, the family $\ec{\FIN}$ is also a proper subset of the family $\cEC(\STAR)$. 
\end{proof}

\begin{lemma}\label{lemma:ec:star_ss_rcom}
  The family $\ec{\STAR}$ is a proper subset of the families $\ec{\LCOM}$ and $\cEC(\RCOM)$.
\end{lemma}
\begin{proof}
The inclusions~$\STAR\setminus \{\{\lambda\}\} \subseteq \LCOM$ and~$\STAR\setminus \{\{\lambda\}\} \subseteq \RCOM$ hold
as recalled in Section~\ref{sec:subreg}. 
Consider an external contextual grammar with a single selection component $(\{\lambda\},C)$ (if there are more components
with the selection language $\{\lambda\}$, they can be joined to one where the new set of contexts is the union 
of the single sets and the selection language is still the same). If the generated language contains the empty word, then this
is an axiom since it cannot be obtained by derivation. Then, exactly the (finitely many) words $uv$ with $(u,v)\in C$ are 
generated using this selection component. Thus, if we put all these words $uv$ with $(u,v)\in C$ into the set of
axioms as well, we can remove the component~$(\{\lambda\},C)$ and obtain a contextual grammar which generates the same
language but has no selection language $\{\lambda\}$ anymore. Then, the remaining selection languages belong to the 
families~$\LCOM$ and~$\RCOM$. Hence, every language of $\ec{\STAR}$ also belongs to the families $\ec{\LCOM}$ and $\cEC(\RCOM)$.

According to Lemma~\ref{lemma:ec:comb_o_pre_star}, the language 
$L = \set{b^na}{n \geq 0}\cup \{\lambda\}$ belongs to $\ec{\COMB}$ (and also to~$\ec{\LCOM}$ and $\ec{\RCOM}$ by 
Theorem~\ref{theorem:neue_hierarchie}, Figure~\ref{fig:lang_erg_1}, and Lemma~\ref{lemma:ec_monoton}) but not to $\ec{\STAR}$. 
This proves the properness of the inclusion.
\end{proof}

\begin{lemma}\label{lemma:ec:def_ss_sydef}
  The family $\ec{\DEF}$ is a proper subset of the family $\ec{\SYDEF}$.
\end{lemma}
\begin{proof}
  Let $G=(V,\{(S_1,C_1),\ldots, (S_n,C_n)\},A)$ be a contextual grammar where 
  all selection languages are definite: $S_i=U_i^*B_i\cup A_i$ for $1\leq i\leq n$.
  We first separate the finite parts and obtain the contextual 
  grammar~$G'=(V,\{(U_1^*B_1,C_1),(A_1,C_1),\ldots, (U_n^*B_n,C_n),(A_n,C_n)\},A)$
  which generates the same language as $G$. Next, we eliminate the components with finite selection languages:
  If a set $B_i$ is empty, then the entire selection language is empty and cannot be used for derivation. Hence,
  we can simply omit such selection components without changing the generated language.
  For every component $(A_i,C_i)$ where $A_i$ is a finite language ($1\leq i\leq n$), 
  we move all words $uwv$ with $(u,v)\in C_i$ and $w\in A_i\cap L(G)$ into the set of axioms. These are finitely many
  (as $A_i$ and $C_i$ are finite) and are exactly the words generated by these components). Hence, we can remove these
  components afterwards. Then, we have obtained a contextual grammar which still generates the same language $L(G)$ but
  has only symmetric definite languages left.
  
  The language $L= \set{a^n}{n \geq 1}\cup \set{ba^nb}{n\geq 1}\cup \set{cba^{2n}bc}{n \geq 1}$
  is a witness language for the properness of the inclusion which,
  according to Lemma~\ref{lem:ec:sydef_o_nc},
  belongs to the family $\ec{\SYDEF}$ but not to the family $\ec{\NC}$ and, hence, not to 
  (since $\cEC(\DEF)\subset\cEC(\NC)$ according to \cite{Dassow.2005}).
\end{proof}

\begin{lemma}\label{lemma:ec:sydef_ss_ps}
  The family $\ec{\SYDEF}$ is a proper subset of the family $\ec{\PS}$.
\end{lemma}
\begin{proof}
  From \cite{Olejar_Szabari.2023}, we know the inclusion $\SYDEF\subseteq\PS$. Therefore, by Lemma~\ref{lemma:ec_monoton},
  we have the inclusion~$\ec{\SYDEF}\subseteq\ec{\PS}$. Its properness follows from Lemma~\ref{lem:ec:ord_o_sydef}
  with~$L = \{a,b\}^*\cup \{c\}\{ab\}^*\{c\}$ which belongs to the family $\ec{\ORD}$
  (and also to $\ec{\PS}$ by \cite{Truthe.2021}) but not to the family $\ec{\SYDEF}$.
\end{proof}

Now, we prove the incomparability relations mentioned in Figure~\ref{fig:lang_erg_2} which have not been proved
earlier. These are the relations regarding the families $\cEC(\STAR)$ and $\cEC(\SYDEF)$ since the families~$\cEC(\LCOM)$, $\cEC(\RCOM)$, and $\cEC(\TCOM)$ coincide with $\cEC(\REG)$ and are therefore not incomparable to the other 
families mentioned.

\begin{lemma}\label{lemma:ec:unv_star_comb_def_ord_nc_ps}
  Let $\cF = \{\COMB, \DEF, \SYDEF, \ORD, \NC, \PS\}$. The family $\ec{\STAR}$ is incomparable to each family $\ec{F}$ 
  with $F \in \cF$.
\end{lemma}
\begin{proof}
  Due to the inclusion relations, it suffices to show that there are two languages $L_1$ and $L_2$ 
  with the properties $L_1\in\cEC(\COMB)\setminus\cEC(\STAR)$ and $L_2\in\cEC(\STAR)\setminus\cEC(\PS)$.
  From Lemma \ref{lemma:ec:comb_o_pre_star}, we have the language $L_1 = \set{b^na}{n \geq 0}\cup \{\lambda\}$.
  From Lemma \ref{lemma:ec:star_o_ps}, we have $L_2 = \set{ba^{2n}b}{n \geq 1} \cup \set{a^n}{n \geq 2}$.
\end{proof}

\begin{lemma}\label{lemma:ec:unv_star_nil_comm_circ}
  Let $\cF = \{\NIL, \COMM, \CIRC\}$. The family $\ec{\STAR}$ is incomparable to each family $\ec{F}$ with $F \in \cF$.
\end{lemma}
\begin{proof}
  Due to the inclusion relations, it suffices to show that there are two languages $L_1$ and $L_2$ 
  with the properties $L_1\in\cEC(\NIL)\setminus\cEC(\STAR)$ and $L_2\in\cEC(\STAR)\setminus\cEC(\CIRC)$.
  From Lemma \ref{lemma:ec:nil_o_star}, we have the language $L_1 = \set{a^nbbb}{n\geq 1} \cup \{\lambda\}$.
  From Lemma \ref{lem:ec:star_o_circ}, we have $L_2 = \set{a^nb^n}{n \geq 1} \cup \set{b^na^n}{n \geq 1}$.
\end{proof}

\begin{lemma}\label{lemma:ec:unv_star_suf}
  The language family $\ec{\STAR}$ is incomparable to the family $\ec{\SUF}$.
\end{lemma}
\begin{proof}
  We have 
  $L_1 = \sets{a,b}^*\set{a^n b^m}{n\geq 1,\ m\geq 1}\cup\set{ca^n b^mc}{n\geq 1,\ m\geq 1}\in\ec{\SUF}\setminus\ec{\STAR}$
  from Lemma \ref{lemma:ec:suf_o_star}.
  From Lemma \ref{lemma:ec:star_o_ps}, we know that 
  $L_2 = \set{a^n}{n \geq 2}\cup \set{ba^{2n}b}{n \geq 1}$ belongs to the family $\ec{\STAR}$ but not to
  $\ec{\PS}$ (and neither to $\ec{\SUF}$ by \cite{Truthe.2021}).
\end{proof}

\begin{lemma}\label{lemma:ec:unv_sydef_suf}
  The language family $\ec{\SYDEF}$ is incomparable to the family $\ec{\SUF}$.
\end{lemma}
\begin{proof}
  We have 
  $L_1 = \{a,b\}^*\cup \{c\}\{\lambda, b\}\{ab\}^*\{c\}\in\ec{\SUF}\setminus\ec{\SYDEF}$
  from Lemma \ref{lem:ec:suf_o_sydef}.
  From~\cite{Dassow.2005}, we know that 
  $L_2 = \set{ab^n}{n\geq 1}\cup\{\lambda\}$ belongs to the family $\ec{\COMB}$ 
  but not to $\ec{\SUF}$. By \cite{Truthe.2021} and Lemma~\ref{lemma:ec:def_ss_sydef}, the language $L_2$
  also belongs to $\cEC(\SYDEF)$.
\end{proof}

\begin{lemma}\label{lemma:ec:unv_sydef_ord_nc}
  The family $\ec{\SYDEF}$ is incomparable to each of the families $\cEC(\ORD)$ and $\cEC(\NC)$.
\end{lemma}
\begin{proof}
  Due to the inclusion relations, it suffices to show that there are two languages $L_1$ and $L_2$ 
  with the properties $L_1\in\cEC(\ORD)\setminus\cEC(\SYDEF)$ and $L_2\in\cEC(\SYDEF)\setminus\cEC(\NC)$.
  From Lemma \ref{lem:ec:ord_o_sydef}, we have~$L_1 = \{a,b\}^*\cup \{c\}\{ab\}^*\{c\}$.
  As $L_2$, we take 
  $L_2 = \set{a^n}{n \geq 1}\cup \set{ba^nb}{n\geq 1}\cup\set{cba^{2n}bc}{n \geq 1}$ from Lemma \ref{lem:ec:sydef_o_nc}.
\end{proof}

\begin{lemma}\label{lemma:ec:unv_sydef_comm_circ}
  The family $\ec{\SYDEF}$ is incomparable to each of the families $\cEC(\COMM)$ and $\cEC(\CIRC)$.
\end{lemma}
\begin{proof}
  Due to the inclusion relations, it suffices to show that there are two languages $L_1$ and $L_2$ 
  with the properties $L_1\in\cEC(\COMM)\setminus\cEC(\SYDEF)$ and $L_2\in\cEC(\SYDEF)\setminus\cEC(\CIRC)$.
  In \cite{Truthe.2021}, it was proved that the language $L_1 = \set{a^n}{n \geq 2}\cup \set{ba^{2n}b}{n \geq 1}$
  belongs to $\cEC(COMM)$ but not to $\cEC(\PS)$ (this can be seen also in the proof of Lemma~\ref{lemma:ec:star_o_ps}).
  By Lemma~\ref{lemma:ec:sydef_ss_ps}, the language $L_1$ neither belongs to the family~$\cEC(\SYDEF)$.
  
  In \cite{Dassow.2005}, it was proved that the language $L_2 = \set{abc^n}{n\geq 1}\cup\set{c^nab}{n\geq 1}$
  belongs $\cEC(\COMB)$ but not to $\cEC(\CIRC)$. By \cite{Truthe.2021} and Lemma~\ref{lemma:ec:def_ss_sydef},
  the language $L_2$ also belongs to the family $\cEC(\SYDEF)$.
\end{proof}

\begin{theorem}[Hierarchy of the $\cEC$ language families]\label{theo:neu}
  The inclusion relations presented in Figure~\ref{fig:lang_erg_2} hold. An arrow from an entry $X$ to
  an entry~$Y$ depicts the proper inclusion $X \subset Y$; if two families are not connected by a directed
  path, then they are incomparable.
\end{theorem}
\begin{proof}
  An edge label refers to the paper or lemma in the present paper where the proper inclusion is shown.
  The incomparability results are proved in Lemmas~\ref{lemma:ec:unv_star_comb_def_ord_nc_ps} 
  to~\ref{lemma:ec:unv_sydef_comm_circ}.
\end{proof}

\section{Conclusion and future work}

In this paper, we have extended the previous hierarchy of subregular language families and 
families generated by external contextual grammars with selection in certain subregular language families.


Various other subregular language families have also been investigated in the past 
(for instance, in~\cite{Bordihn_Holzer_Kutrib.2009, Han_Salomaa.2009, Olejar_Szabari.2023}). 
Future research will be on extending and unifying current hierarchies of
subregular language families 
(presented, for instance, in \cite{Dassow_Truthe.2023,Truthe.2021})
by additional families and to use them as control in external contextual grammars. 
We already started investigations on the position of prefix- and infix-closed as well as
prefix-, suffix-, and infix-free languages in the current hierarchy and their impact on the
generative power of external contextual grammars when used for selection. 
The extension of the hierarchy with other families of definite-like languages (for instance,
ultimate definite, central definite, noninital definite) has also already begun. 

The research will be also extended to internal contextual grammars or tree-controlled grammars 
where results are already available in \cite{Dassow_Truthe.2023, Truthe.2021, Truthe.2023.ncma, Truthe.2023.afl}.

\bibliographystyle{eptcs}
\bibliography{generic}
\end{document}